\documentclass[11pt,letterpaper]{article}
\usepackage{amsfonts,amsmath,amsthm}
\usepackage[left=1in, right=1in, top=1in, bottom=1in]{geometry}

\usepackage[doublespacing]{setspace}

\usepackage{lscape}          %landscape single pages
\usepackage{afterpage}
\usepackage{comment}
\usepackage{tikz}
\usetikzlibrary{positioning}

\usepackage{threeparttable} 

\usepackage{titlesec}
\titleformat*{\section}{\large\bfseries}
\titleformat*{\subsection}{\normalfont\bfseries}
\titleformat*{\subsubsection}{\normalfont\itshape\bfseries}

\usepackage{graphicx}
\usepackage{subcaption}
\usepackage{float}
\usepackage{booktabs}
\usepackage{bm}
\usepackage{bbm}
\usepackage{multirow}
\usepackage{xcolor}
\usepackage{makecell}
\usepackage{MnSymbol}
\usepackage[many]{tcolorbox} 
\usepackage{authblk} 

\usepackage{subcaption}

\usepackage{ulem}
\newcommand\bluesout{\bgroup\markoverwith{\textcolor{blue}{\rule[0.5ex]{2pt}{0.8pt}}}\ULon}

\usepackage[font={footnotesize,bf}]{caption}

\usepackage[sort]{natbib}

\makeatletter
\renewenvironment{proof}[1][\proofname]{%
	\par
	\pushQED{\qed}%
	\normalfont \topsep6\p@\@plus6\p@\relax
	\trivlist
	\item[\hskip\labelsep
	\small\bfseries
	#1\@addpunct{.}]\ignorespaces
}{%
	\popQED\endtrivlist\@endpefalse
}
\makeatother

\usepackage[colorlinks=true,linkcolor=blue,citecolor=blue,urlcolor=blue,breaklinks=True]{hyperref}%
\usepackage[font={footnotesize,bf}]{caption}

\usepackage{amsthm}
\numberwithin{equation}{section}
\newtheorem{theorem}{Theorem}

\newtheorem{lemma}{Lemma}

\newtheorem{proposition}{Proposition}

\begin{document}
	
	\title{\centering 
		{Agentic AI and Hallucinations}
	}
	\bigskip 
	\bigskip
	\bigskip
	\author[,1]{Engin Iyidogan\thanks{Corresponding author.\\ Email addresses:
			\url{engin.iyidogan@skema.edu} (E. Iyidogan),
			\url{ali.ozkes@skema.edu} (Ali I.~Ozkes)}}
	\author[1,2]{Ali I.~Ozkes}
	
	\affil[1]{SKEMA Business School – Université C\^ote d’Azur, Paris, France}
	\affil[2]{WU Vienna University of Economics and Business, Vienna, Austria}

	\bigskip
	\bigskip
	\date{\today}
	
	\maketitle
	
	\thispagestyle{empty}
	\bigskip
	
	\begin{abstract}
		
		We model a competitive market where AI agents buy answers from upstream generative models and resell them to users who differ in how much they value accuracy and in how much they fear hallucinations. Agents can privately exert effort for costly verification to lower hallucination risks. Since interactions halt in the event of a hallucination, the threat of losing future rents disciplines effort. A unique reputational equilibrium exists under nontrivial discounting. The equilibrium effort, and thus the price, increases with the share of users who have high accuracy concerns, implying that hallucination-sensitive sectors, such as law and medicine, endogenously lead to more serious verification efforts in agentic AI markets.

	\end{abstract}

	\vfill
	
	
	\newpage

	\section{Introduction}
	\label{sec:introduction}

	Many industries increasingly rely on AI agents \citep{ftAgentsFrom, brynjolfsson2025generative}. These software entities intermediate between upstream general-purpose model providers (\textit{e.g.}, large language models, LLMs)  and end users by converting queries into structured prompts, routing them to one or several models, and post‑processing the returned completions \citep{rothschild2025agentic}. Their economic role resembles that of traditional data vendors and consultancy firms, yet two technological facets distinguish them. First, the upstream supply side offers a discrete menu of models that differ markedly in both their per‑query price and their output quality, measured for instance by the baseline hallucination rate, the rate at which the model falsely claims to have performed an action or generated a correct response \citep{dahl2024large, canayaz2025ai}. Second, downstream agents can apply various techniques to reduce hallucinations, such as retrieval‑augmentation, ensembling, or human‑in‑the‑loop verification \citep{shavit2023practices}. In other words, AI agents choose two quality levers: which model to call and how much effort to expend in verification.

	Agents' quality choices determine the equilibrium mapping from retail price to delivered reliability. When hallucination risk is not perfectly internalized, the standard welfare theorems break down. In high‐stakes domains such as law or medicine, this becomes vitally important \citep{asgari2025framework}. For such users, the disutility from an incorrect answer exceeds by orders of magnitude the service’s sticker price. Users with lower sensitivity to such failures, such as those who use the service for entertainment, may rather treat hallucinations as a nuisance. This raises the question: Can competition and relational contracting alone induce AI agents to verify answers to the standard demanded by high-criticality users, and how does the equilibrium adjust to the mix of high- versus low-criticality demand?

	We develop a tractable discrete-time, infinitely-repeated model in which competitive AI agents pick the wholesale model, verification effort, and retail price. Users differ both in the value they place on a correct answer and in their aversion to hallucinations. Relational contracts make costly verification incentive-compatible: the threat of losing future rents disciplines the agent and resolves the moral-hazard problem. Our model yields testable predictions on how market composition shapes agentic AI service quality and pricing.
	
	The verification mechanism we study is akin to relational‑incentive contracts \citep{levin2003relational} in that costly, non‑contractible effort is enforced by the threat of future lost surplus. In our study, that bilateral logic is embedded in a competitive market whose users differ in their tolerance for hallucinations, so the agent’s informational rent is passed through to prices. Our paper also complements the allocation‑of‑authority framework by  \citet{athey2020allocation}. While they study whether a principal should delegate decision rights to an AI or retain them in‑house, we take AI authority as given and analyze the verification effort exerted by a profit‑maximizing operator.

	\section{The Model}
	\label{sec:model}
	We develop a dynamic economy that features \textit{(i)} a unit mass of users with heterogeneous valuation and hallucination aversion, \textit{(ii)} a perfectly competitive continuum of (AI) agents who assemble information services, and \textit{(iii)} a finite set of upstream providers of models (\textit{e.g.}, LLMs). We model the interaction as an infinitely repeated game where uncertainty arises from the stochastic occurrence of hallucinations.
	
	More precisely, users and agents interact in a competitive market, choosing from a menu of services offered by model providers. Time is discrete, and all parties share a common discount factor $\delta \in [0,1)$. Each user is endowed with a privately observed type $\theta\in\{H\text{(igh)},L\text{(ow)}\}$, which governs the utility $v(\theta)>0$ from a correct answer and the disutility $\alpha(\theta)>0$ from a hallucinated answer. We impose $v(H)>v(L)$ and, crucially, $\alpha(H)>\alpha(L)$ to capture the heightened stakes of the segment with high-type users. The population share of $H$ is $\mu\in(0,1)$.
	
	Each agent is a long-lived firm that chooses a model $m\in\mathcal{M}$, where $\mathcal{M}$ is publicly known to be finite, to query and a retail price $p$ to post. Although $m$ and $p$ stay the same throughout, the agent chooses its verification effort level in each period, $e\in\mathbb{R}_{\ge0}$. Model $m$ carries a wholesale fee $k_m>0$ per call and a baseline hallucination probability $h_0(m)\in(0,1)$. We assume $m\neq m'$ implies $k_m\neq k_{m'}$ and $h_0(m)\neq h_0(m')$.
	
	Verification effort, chosen by the agent when interacting with a user, lowers the hallucination probability exponentially:
	\begin{equation}\label{eq:halprob}
		h(m,e)\;=\;h_0(m)\,\exp(-\beta e),
	\end{equation}
	where $\beta>0$ denotes verification efficacy. Effort induces a convex cost $c(e)$ satisfying $c(0)=0$, $c'(e)>0$, and $c''(e)>0$ for all $e>0$.
	
	In any interaction, agents first post long-term relational contracts, specified by a model $m$ and a price per-period $p$. Users observe the available contracts and choose one to enter a long-term relationship. Then, in each subsequent period, the user pays $p$ and the agent chooses an unobservable effort $e$. After the service is delivered, a hallucination may occur with probability $h(m, e)$. If a hallucination occurs, the user observes this poor outcome, the relationship terminates, and the user exits the market. If no hallucination occurs, the relationship continues to the next period. In each period, a mass of new, uninformed users enters the market, equal to the mass of users who exited in the previous period, keeping the total population constant.
	
	A user of type $\theta$ receives expected per-period utility 
	\[
	U_{\theta}(p,m,e) = (1-h(m,e))v(\theta) - h(m,e)\alpha(\theta) - p.
	\]
	The total lifetime value to the user from entering the relationship is the expected discounted sum of these per-period utilities:
	\begin{equation}
		V_{\theta}(p,m,e) = \frac{U_{\theta}(p,m,e)}{1-\delta(1-h(m,e))}.
	\end{equation}
	The user commits to the relationship if and only if their expected lifetime value is non-negative, $V_{\theta}(p,m,e) \ge 0$.
	
	We analyze equilibria where agents offer a single, unified contract. In equilibrium, this contract must be self-sustaining for the agent and acceptable to all participating user types. An agent's commitment to effort $e$ is sustained by the value of the ongoing relationship. The incentive compatibility constraint is then
	\begin{equation}\label{eq:ic_constraint}
		c(e) \le \delta[h(m,0) - h(m,e)] V_C,
	\end{equation}
	where $V_C$ is the agent's present value of relational rents, such that
	\[V_C = \frac{p - k_m - c(e)}{1-\delta(1-h(m,e))}.\] 
	Competition leads agents to lower the price $p$ until this constraint binds, which yields the minimum price required to sustain effort $e$:
	\begin{equation}\label{eq:price_enforce}
		p(m,e) = k_m + c(e) \left[ 1 + R(m,e) \right],
	\end{equation}
	where $$R(m,e)= \frac{1 - \delta(1-h(m,e))}{\delta(h(m,0)-h(m,e))},$$ such that $c(e)R(m,e)$ denotes the agent's per-period markup (or rent). An equilibrium contract is one that maximizes average user lifetime value, subject to the enforcement price \eqref{eq:price_enforce} and the participation constraint, $V_\theta \ge 0$, for $\theta\in\{H,L\}$.
	
	% ------------------------------------------------------------
	% ------------------------------------------------------------
	\section{Results}
	\label{sec:reputation}
	We analyze the market outcomes under relational contracting, focusing on how the value of reputation enables verification effort in a market with heterogeneous users.
	
	\subsection{Benchmark: The Spot Market Outcome}
	We first consider the benchmark where future interactions have no value, \textit{i.e.}, $\delta=0$. For brevity, we let
	\( U_H(m,p,e) \equiv U_{\theta=H}(p,m,e) \) (and \( U_L(m,p,e)\) analogously).

	\begin{proposition}
		If $\delta=0$, the only sustainable effort level is $e^*=0$. Competitive agents offer the contract $(m^*, p^*)$ where $p^*=k_{m^*}$ and the model $m^*$ is the one that maximizes the $\mu$-weighted average of user utilities:
		\[
		m^* = \arg\max_{m \in \mathcal{M}} \left\{ \mu U_H(m, k_m, 0) + (1-\mu)U_L(m, k_m, 0) \right\}
		\]
		The market is active only if this contract is acceptable to both user types.
	\end{proposition}
	\begin{proof}
		When $\delta=0$, the IC constraint \eqref{eq:ic_constraint} implies $e^*=0$. An agent's cost thus equals $k_m$, and competition forces the price $p^*=k_m$. To attract users, agents must offer a contract that provides the highest possible average utility to the market, thus solving for $m^*$ as defined in the proposition.
	\end{proof}
	This benchmark highlights that without the shadow of the future, agents have no incentive to exert costly verification effort, and the quality of service provision regresses to its minimum possible level, $e=0$. Competition can only occur over the observable dimensions of baseline model quality $h_0(m)$ and its associated cost $k_m$, demonstrating that a mechanism like reputation is essential to solving the underlying moral hazard problem.

	% ------------------------------------------------------------
	% ------------------------------------------------------------
	\subsection{The Reputational Equilibrium}
	When $\delta > 0$, agents can exert positive effort. 
	An equilibrium contract must maximize the average user's lifetime value, subject to the enforcement price (\ref{eq:price_enforce}). This is equivalent to solving a planner's problem that maximizes the total surplus allocated to users, net of the incentive rents required by the agent ($c(e)R(m,e)$ discounted). Let $v_{avg}(\mu) = \mu v(H) + (1-\mu) v(L)$ and $\alpha_{avg}(\mu) = \mu \alpha(H) + (1-\mu) \alpha(L)$. The participation constraint becomes $V_\theta(m,e)\ge0$, where $$V_\theta(m,e) =\frac{(1-h(m,e))v(\theta)-h(m,e)\alpha(\theta)-p(m,e)}{1-\delta(1-h(m,e))},$$
	for $\theta\in\{H,L\}$. The objective is:
	\begin{equation}\label{eq:welfare_objective_general}
		\max_{m,e\ge0} \quad W(m,e) = \frac{(1-h(m,e))v_{avg}(\mu)-h(m,e)\alpha_{avg}(\mu)-k_m-c(e)}{1-\delta(1-h(m,e))} - \frac{c(e)R(m,e)}{1-\delta(1-h(m,e))},
	\end{equation}
	subject to participation constraints. To solve this, we first identify which constraint is binding in relation to the equilibrium hallucination probability $h^*(m,e)$. 
	
	\begin{lemma}
		\label{lem:binding_general}
		Let $\kappa = \frac{\alpha(H)-\alpha(L)}{v(H)-v(L)}$ be the sensitivity ratio.
		\begin{enumerate}
			\item[(i)] $h^*(m,e) > {1}/({1+\kappa}) \implies V_H(p,m,e) < V_L(p,m,e)$, thus, $V_H \ge 0$ is binding.
			\item[(ii)] $h^*(m,e) < {1}/({1+\kappa}) \implies V_H(p,m,e) > V_L(p,m,e)$, thus, $V_L \ge 0$ is binding.
		\end{enumerate}
	\end{lemma}
	\begin{proof}
		The ordering of $V_H$ and $V_L$ is determined by the ordering of the per-period utilities $U_H$ and $U_L$. We have  $U_L \ge U_H$ if and only if 
		\[h(m,e)[\alpha(H)-\alpha(L)] \ge (1-h(m,e))[v(H)-v(L)],\] 
		which yields the threshold ${1}/({1+\kappa})$. The two constraints coincide when $h^*(m,e) = {1}/({1+\kappa})$.
	\end{proof}
	
	This leads to our main result stated in Theorem \ref{prop:reputation_market_general}, which characterizes the equilibrium for nontrivial discount factors. We first note that an active equilibrium requires that the participation constraint for the binding users, \textit{i.e.,} the $L$ types, is met. This requires that the per-period utility is sufficient to cover the agent's rent, \textit{i.e.}, $(1-h)v_L - h\alpha_L - k_m - c(e) \ge c(e)R(m,e)$. Substituting the expression for the rent term $R(m,e)$ and solving for $\delta$ yields the condition $\delta \ge \underline{\delta}(m,e)$, where \[\underline{\delta}(m,e) := \left( (1-h) + (h_0-h) \frac{(1-h)v_L-h\alpha_L-k_m-c(e)}{c(e)} \right)^{-1}.\] Thus, an active market is only sustainable if players are sufficiently patient.

	\begin{theorem}
		\label{prop:reputation_market_general}
An active relational equilibrium contract exists if and only if $\delta \ge \underline{\delta}(m^*, e^*)$, where the equilibrium is characterized by the unique pair $(m^* , e^* )$. The equilibrium verification effort $e^*$ is strictly increasing in the share of high-type users, $\mu$.	
\end{theorem}

	\begin{proof}
		The equilibrium $(m^*, e^*)$ solves \ref{eq:welfare_objective_general} whenever $\delta \ge \underline{\delta}(m^*, e^*)$. For given $m$, the first-order condition defines an interior solution $e^*(m)$ and precisely formulates the trade-off between generated surplus and agent's markup such that
		\begin{equation}
			\frac{d}{de}\left( \frac{(1-h)v_{avg}(\mu)-h\alpha_{avg}(\mu)-k_m-c(e)}{1-\delta(1-h)} \right) = \frac{d}{de}\left( \frac{c(e)}{\delta(h_0-h)} \right),
		\end{equation}
		where $h=h(m,e)$ and $h_0=h_0(m)$. Uniqueness is guaranteed if $\partial^2 W/\partial e^2 < 0$, which holds if $c''(e)$ is sufficiently large. Differentiating the first-order condition with respect to $\mu$, we have
		$$\frac{de^*}{d\mu}=-\frac{ {\partial^2 W}/{\partial e \partial \mu} }{{\partial^2 W}/{\partial e^2 }}.$$ The denominator is negative by the second-order condition. The cross-partial derivative becomes:
		\[
		\frac{\partial^2 W}{\partial e \partial \mu} = \frac{\beta h \left( (v_H-v_L) + (1-\delta)(\alpha_H-\alpha_L) \right)}{\big(1-\delta(1-h)\big)^2},
		\]
		which is positive as $v_H > v_L$, $\alpha_H > \alpha_L$, $\beta, h > 0$, and $\delta \in (0,1)$. It follows that ${de^*}/{d\mu} > 0.$
	\end{proof}

	\noindent Our main contributions are twofold. First, the model highlights the trade-off between a model's baseline hallucination risk, $h_0(m)$, and an agent's verification effort, $e$, made visible in both the pricing and equilibrium conditions. Second, the comparative static finding $de^*/d\mu>0$ delivers a sharp exposition: market composition itself acts as a disciplinary device, in that sectors with higher composition of high-criticality users, such as law and medicine, endogenously generate more heavily verified and consequently more expensive AI services.

	% ------------------------------------------------------------
	\subsection{Numerical Example}
	\label{sec:numerics}
	
	We conclude our analysis with a numerical example. Low-type users gain \(v(L)=1\) from a correct answer and incur \(\alpha(L)=1.5\) if the answer is hallucinated; high-type users gain \(v(H)=3\) and lose \(\alpha(H)=10\). Verification lowers hallucination risk according to \eqref{eq:halprob} with efficacy \(\beta=0.70\) and incurs quadratic cost \(c(e)=\frac{e^{2}}{8}\).  
	The discount factor is set to \(\delta=0.95\). The baseline upstream model (A) features a hallucination probability \(h^A_{0}=0.20\) and wholesale fee \(k_A=0.05\). A second model, Model B, is such that \(h^B_{0}=0.13\) and \(k_B=0.30\).

	Figure~\ref{fig:effort_delta} shows that greater patience relaxes the incentive-compatibility constraint, shifting the whole \(e^{*}(\mu)\) schedule upward, while effort is strictly increasing in \(\mu\) as established in Theorem~\ref{prop:reputation_market_general}.  Verification plateaus once \(e^{*}\approx1.97\), where the hallucination rate falls below 0.02.

	\begin{figure}[h]
		\centering
		% ---------- Panel (a)
		\begin{subfigure}[t]{0.48\textwidth}
			\centering
			\includegraphics[width=\textwidth]{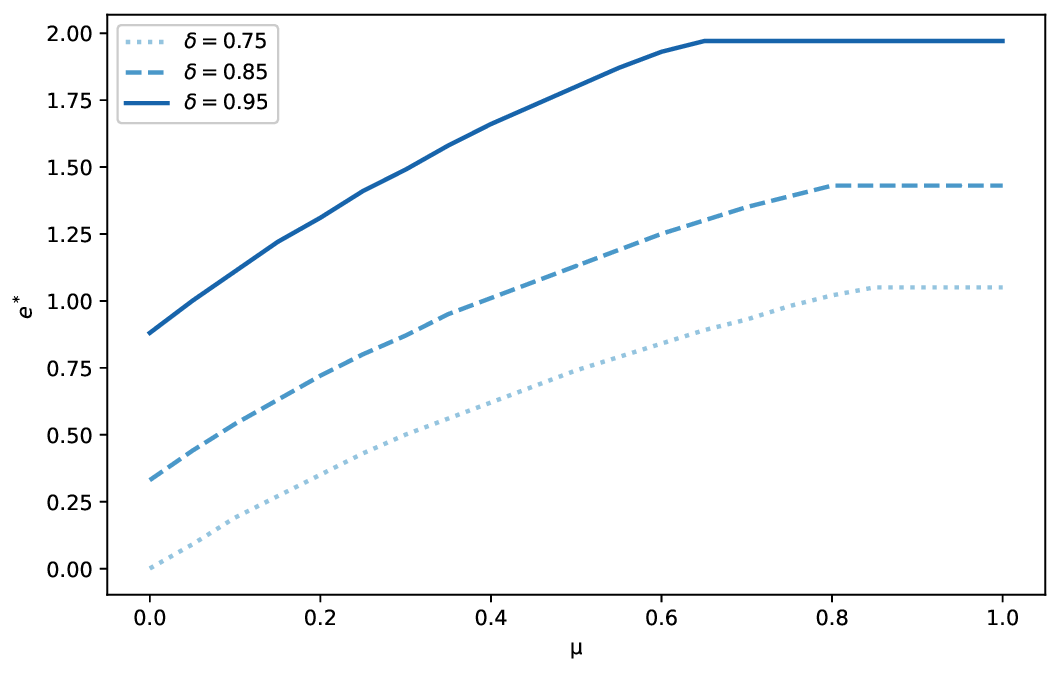}
			\caption{Patience and verification effort.%
				\newline\small Equilibrium effort \(e^{*}\) as a function of
				the share of high-type users \(\mu\) for three discount
				factors (\(\delta=0.75,0.85,0.95\)).}
			\label{fig:effort_delta}
		\end{subfigure}
		\hfill
		% ---------- Panel (b)
		\begin{subfigure}[t]{0.48\textwidth}
			\centering
			\includegraphics[width=\textwidth]{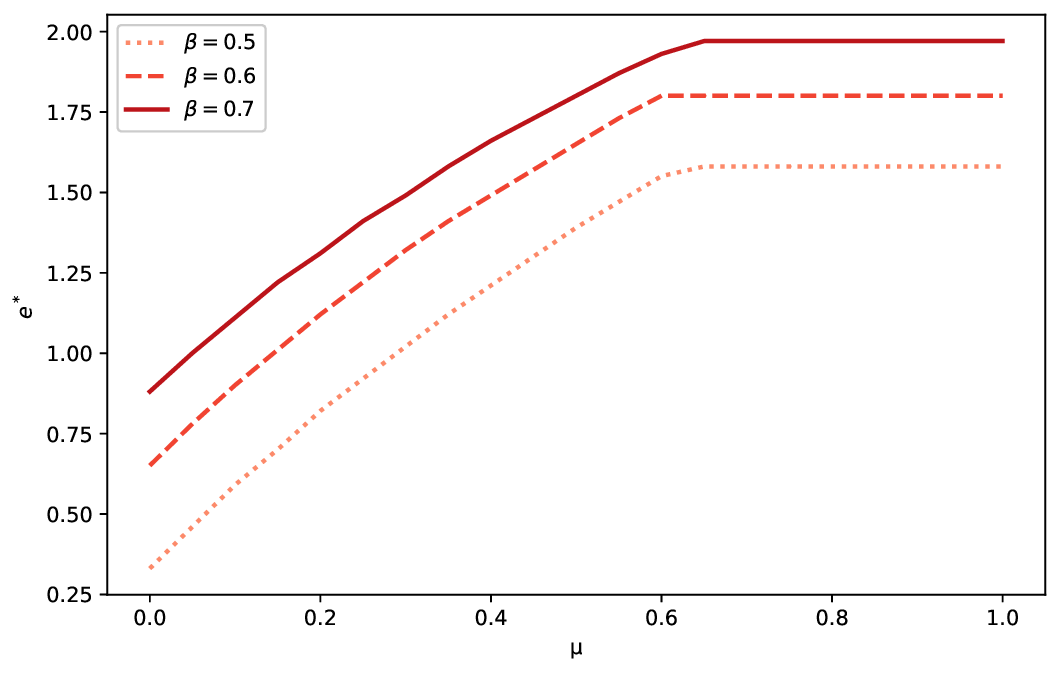}
			\caption{Verification-technology comparative static.%
				\newline\small Equilibrium effort \(e^{*}\) as a function of
				\(\mu\) for three verification efficacies
				(\(\beta=0.50,0.60,0.70\)).}
			\label{fig:effort_beta}
		\end{subfigure}
		
		\caption{Equilibrium verification effort under varying discount
			factors and verification efficacies.}
		\label{fig:effort_composite}
	\end{figure}

	Figure~\ref{fig:effort_beta} shows that more effective screening technology lifts the \(e^{*}(\mu)\) schedule, allowing agents to reach the quality ceiling at lower \(\mu\).  The result follows from the effort efficacy as in \eqref{eq:halprob}.

	% ------------------------------------------------------------

	Figure~\ref{fig:welfare_switch} illustrates the upstream model choice. For \(\mu<0.26\) agents rely on the cheaper but riskier model and compensate with higher verification effort. When the high-type composition exceeds roughly one-fourth of the market, paying the higher wholesale fee for the cleaner model maximizes surplus, and quality provision shifts from the retail to the upstream layer.

	\begin{figure}[h]
		\begin{center}
			\includegraphics[width=0.5\textwidth]{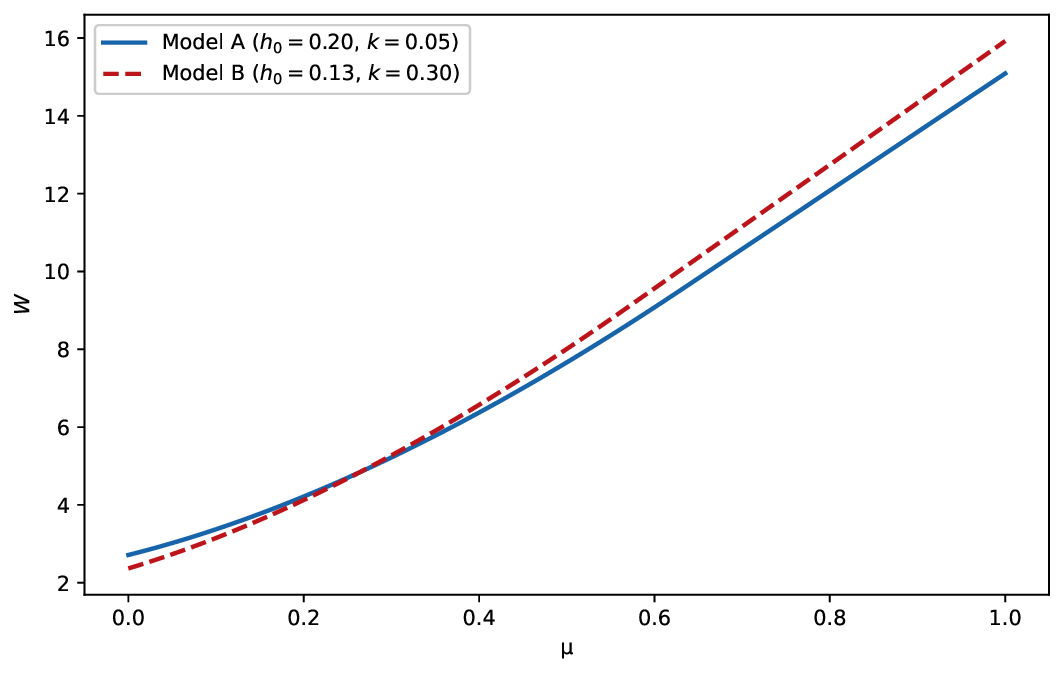}
		\end{center}
		\caption{\textit{Upstream model choice.} Notes: Total welfare \(W\) under Model A (\(h^A_{0}=0.20,k_A=0.05\), solid) and Model B (\(h^B_{0}=0.13,k_B=0.30\), dashed).}
		\label{fig:welfare_switch}
	\end{figure}
	
	% ------------------------------------------------------------
	
	\newpage
	\bibliographystyle{chicago}
	\bibliography{reference_ai_agents}

\end{document}